\documentclass[10pt,conference,a4paper]{IEEEtran}
\usepackage[latin1]{inputenc}
\usepackage[cmex10]{amsmath}
\usepackage{amsfonts}
\usepackage{amssymb}

\usepackage[ruled,vlined,linesnumbered]{algorithm2e}
\usepackage{bm}
\usepackage{booktabs}
\usepackage{caption}
	\usepackage{subcaption}
\usepackage{graphicx}
\usepackage{tikz}
\usepackage{url}

\newtheorem{lemma}{Lemma}
\newtheorem{example}{Example}

\title{Using the distribution of cells by dimension in a cylindrical algebraic decomposition}

\author{David Wilson, Matthew England, Russell Bradford \& James H. Davenport  \\
Department of Computer Science, University of Bath, Bath, BA2 7AY, UK \\
E-mail: \texttt{\{D.J.Wilson, M.England, R.J.Bradford, J.H.Davenport\}@bath.ac.uk}}

\begin{document}

\maketitle

\begin{abstract}
We investigate the distribution of cells by dimension in cylindrical algebraic decompositions (CADs).  We find that they follow a standard distribution which seems largely independent of the underlying problem or CAD algorithm used.  Rather, the distribution is inherent to the cylindrical structure and determined mostly by the number of variables.

This insight is then combined with an algorithm that produces only full-dimensional cells to give an accurate method of predicting the number of cells in a complete CAD.  Since constructing only full-dimensional cells is relatively inexpensive (involving no costly algebraic number calculations) this leads to heuristics for helping with various questions of problem formulation for CAD, such as choosing an optimal variable ordering.  Our experiments demonstrate that this approach can be highly effective.
\end{abstract}

\section{Introduction}
\label{SEC:Intro}

\subsection{Background on CAD}

A \textbf{cylindrical algebraic decomposition} (CAD) is:
\begin{itemize}
\item a \textbf{decomposition} of $\mathbb{R}^n$, meaning a collection of cells which do not intersect and whose union is $\mathbb{R}^n$;
\item \textbf{cylindrical}, meaning the projections of any pair of cells with respect to a given variable ordering are either equal or disjoint;
\item \textbf{(semi)-algebraic}, meaning each cell can be described using a finite sequence of polynomial relations.
\end{itemize}
The first algorithm to produce CADs was introduced by Collins \cite{ACM84I}.  Here we start with a projection phase to derive a set of polynomials relative to the input, then CADs are built incrementally by dimension according to the zeros of those polynomials (a process known as lifting).  

Each cell is represented by a \textbf{cell-index}: an $n$-tuple of integers defining its position in the CAD.  An even integer is referring to a variable taking the value of one of the (ordered) real roots of the projection polynomials and an odd integer means that a variable is within an interval between two of these.  In addition each cell would usually be accompanied by a sample point used in the construction.  

We would normally compute a CAD to solve an underlying problem.  Most notably, it can be a tool for quantifier elimination (QE) over the reals.  Here we must build a CAD relative to a quantified formula such that the Boolean value of that formula is invariant (true or false) in each cell.  Then an equivalent quantifier-free formula may be computed from the semi-algebraic descriptions of the true cells.
CAD has been applied elsewhere, including problems in parametric optimisation \cite{FPM05}, epidemic modelling \cite{BENW06}, theorem proving \cite{Paulson2012}, motion planning \cite{WDEB13} and reasoning with multi-valued functions and their branch cuts \cite{DBEW12}.

The original, and most common CAD algorithm gives output that is \textbf{sign-invariant} with respect to a set of polynomials.  This means that each polynomial in the input has constant sign on each cell of the CAD created.   However, a CAD can be produced more efficiently if we work closer to the underlying problem, for example, by building CADs invariant with respect to the truth of formulae \cite{BDEMW13}, or making use of the structure of the quantifiers \cite{CH91}.  Other important advances in CAD theory include the use of certified numerics \cite{Strzebonski2006, IYAY09} and an alternative approach to Collins' algorithm using the theory of regular chains and triangular decomposition \cite{CMXY09, BCDEMW14}.  

In this paper we are more concerned with a CAD as the mathematical object satisfying the definition above (rather than the particular algorithm which produces one).

\subsection{Contribution}

Each cell in a CAD has a \textbf{dimension} which can range from 0 (when the cell is a point) to $n$ (when the cell is of full-dimension in $\mathbb{R}^n$).  We describe a subset of cells from a CAD as a \textbf{sub-CAD}.  The sub-CAD consisting of only those cells with full-dimension has been a much studied topic \cite{McCallum1993, Strzebonski2000}.  It can be identified far more efficiently than the CAD itself and is sufficient to solve certain classes of problems.  More generally, we define those cells in a CAD with the same dimension as a \textbf{layer} and an $\bm{\ell}$\textbf{-layered sub-CAD} as a sub-CAD of $\mathbb{R}^n$ consisting of those cells with dimensions $d \in [n-\ell+1, n]$ (i.e. the top $\ell$ layers of the CAD).  We call the CAD consisting of all $n+1$ layers the \textbf{complete CAD}.  See Section  2.2 of \cite{WBDE14} for details including algorithms to produce such sub-CADs both directly and recursively (where additional layers are created one at a time).

In Section \ref{SEC:Distribution} we investigate the spread of cell dimensions in a CAD.  We discover that they conform to a common distribution regardless of the problem studied or algorithm used.  Rather the distribution is a feature of the cylindrical structure and determined mostly by the number of variables.  This means that the size (number of cells) of a CAD may be predicted accurately by the number of full-dimensional cells (which can be computed far quicker).  In Section \ref{SEC:Heuristic} we investigate using this as a heuristic for deciding questions of problem formulation for CAD, showing promising experimental results.  We now continue the introduction with a simple  example to illustrate the ideas so far.

\subsection{Motivating Example}

\begin{example}
\label{ex:intro}
Consider $f = x-y^2$ and $g = x^2-y^2-1$;
the polynomials graphed respectively by the circle and parabola in the first image of Fig.~\ref{fig:MVEx}.  
\end{example}
A CAD is defined implicitly with respect to a variable ordering (defining the projections used for cylindricity).  Assume an ordering $y \succ x$ meaning projections are from $(x,y) \rightarrow x$. 
The second image of Fig.~\ref{fig:MVEx} visualises a sign-invariant CAD for $\{f,g\}$ in this ordering by marking each cell with a black box.  If a box lies at the intersection of two curves, including the dotted lines) then they indicate a cell of dimension 0: just that point.  Otherwise, if the box lies on one of the curves or dotted lines then it indicates a cell of dimension 1: that line segment.  The remaining boxes indicates cells of full dimension: portions of $\mathbb{R}^2$ bounded by the curves or dotted lines.  

In fact this is the minimal sign-invariant CAD for $\{f,g\}$ in the ordering, that is, the one with the fewest number of cells which satisfies the definitions. For example, consider $x \in (-1,0)$.  Then the sign-invariant condition means we must distinguish the five cells indicated (the two portions of the circle and the spaces between, above and below). We could not extend these cells beyond $(-1,0)$ without violating the cylindricity or sign-invariance conditions.  Similar arguments show that the 51 cells indicated are indeed the minimum.

The third image in Fig.~\ref{fig:MVEx} shows only the full dimensional cells (as portions of $\mathbb{R}^2$ coloured differently to neighbouring cells).  There are 17 of these.  The question answered affirmatively in this paper is whether the complexity of the second image (number of black squares) can be predicted accurately by the complexity of the third image (number of different coloured portions).  The experiments in Section \ref{SEC:Distribution} suggest that for a CAD in two variables approximately 0.334 of the cells are full dimensional.  Hence knowing the number of full dimensional cells only, we would correctly predict there to be $17/0.334 = 50.898 \simeq 51$ cells in the complete CAD.

\begin{figure}
\caption{
From top to bottom: graphs of $f=x-y^2$ and $g=x^2+y^2-1$; a sign-invariant CAD for $\{f,g\}$, the full dimensional cells in the CAD.  The latter two use variable ordering $y \succ x$.
}
\label{fig:MVEx}
\includegraphics[width=8.8cm]{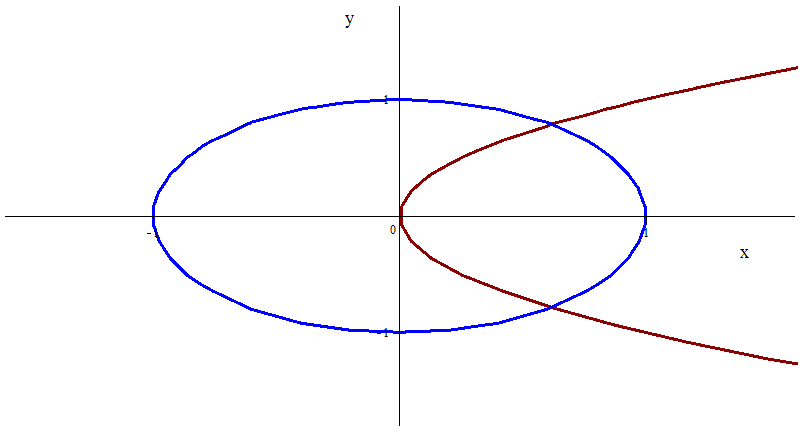}\\
\includegraphics[width=8.8cm]{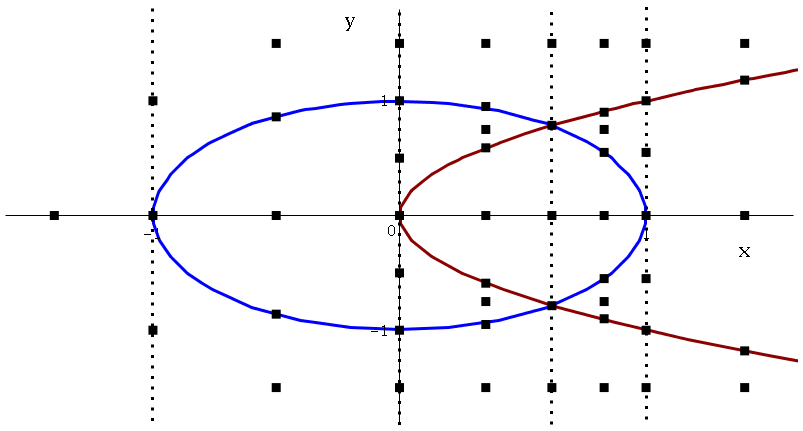}\\
\includegraphics[width=8.8cm]{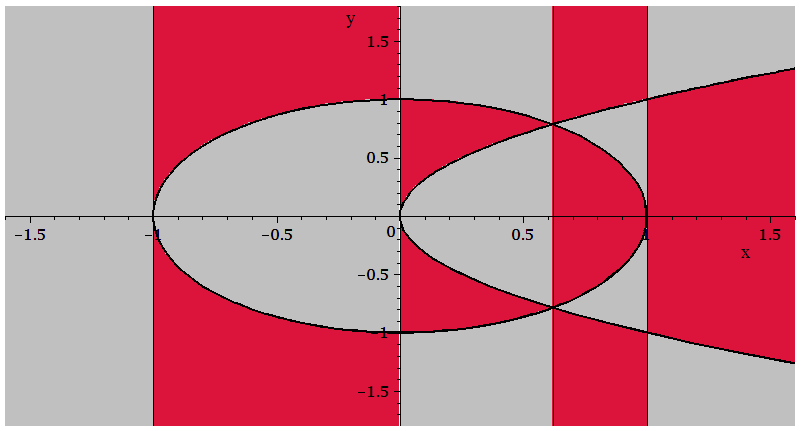}
\end{figure}

\noindent \textbf{Remarks:} We note the following about this example:
\begin{enumerate}
\item Although the full dimensional cells may be sufficient to solve certain problems others will require knowledge of the complete CAD.  For example, the formula $f=0 \land g<0$ is not true on any cell of full dimension.
\item The images above relate to a hypothetical minimal CAD, not necessarily one produced by a known algorithm.  CAD algorithms identify points of intersection by taking resultants.  In this case we have
\[
\mbox{res}_{y}(f,g) = (x^2+x-1)^2
\] 
which has roots at $x = \tfrac{1}{2}(-1 \pm \sqrt{5})$.  The root at $0.618$ identifies the real intersections of $f$ and $g$ while the other at $-1.618$ identifies intersections in $\mathbb{C}^2$ (at points with complex $y$ coordinate). 
Hence, while not required for sign-invariance, known CAD algorithms would split the leftmost cell into three (the line $x=-1.618$ and two full dimensional cells either side).  

The number of full dimensional cells then increases to 18 and the predicted number of total cells becomes $18/0.334 = 53.892 \simeq 54$, one more than the total.
\end{enumerate}

Constructing CADs for the example above is easy with modern technology.  However, for larger problems (particularly those with more variables) CAD can be challenging.  
The cost of computing full-dimensional cells also increases with the size of the problem, but it is much simpler as it avoids computation with algebraic numbers.  Hence, in many situations it is feasible to use the full dimensional cells as a metric to predict the size of the complete CAD, or the feasibility of computing it.  

This approach can be useful for deciding questions of problem formulation for CAD, such as when there is a free or constrained choice over the variable ordering.  For example, when using CAD for QE variables must be projected in the order they are quantified but we can change the ordering of free variables, or variables in blocks of the same quantifier.  The minimal sign-invariant CAD for Example~\ref{ex:intro} with ordering $x \succ y$ has 47 cells, 16 of which are full dimensional as shown in Fig.~\ref{fig:MVEx2}.  If we calculated just these we would predict $16/0.334$
$\simeq 48$ cells in the complete CAD.  Hence using the number of full-dimensional cells as a heuristic would lead us to use the variable ordering producing the smaller complete CAD.

\begin{figure}
\caption{The full dimensional cells in a CAD for $f=x-y^2$ and $g=x^2+y^2-1$, with variable ordering $x \succ y$.}
\label{fig:MVEx2}
\includegraphics[width=8.8cm]{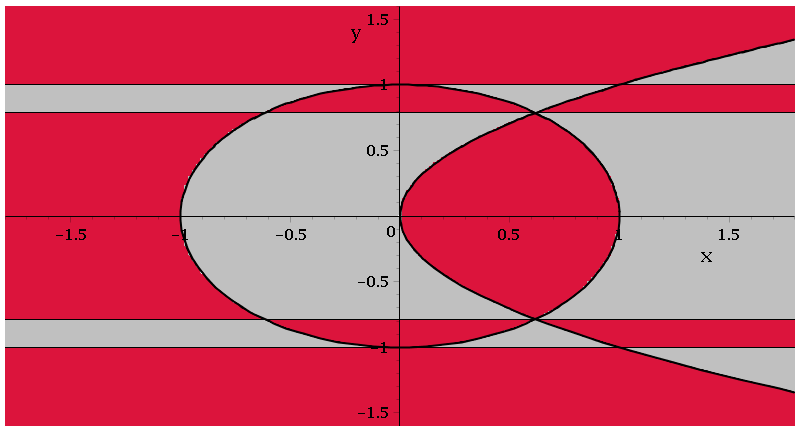}
\end{figure}

There do exist other heuristics for making such choices, which can be cheaper, although they may not be as closely correlated.  See \cite{EBDW14} for a recent summary.  An important example is Brown's heuristic \cite{Brown2004} which uses only simple measures on the input to decide an ordering.  Recent studies show that while it usually gives a good choice, there are classes of problems where it is not successful \cite{HEWDPB14}.  For this example the measures used by this heuristic do not discriminate between $y \succ x$ and $x \succ y$.  
The heuristic \texttt{sotd} \cite{DSS04} goes further by considering the full set of projection polynomials produced by a specific CAD algorithm.  It measures the total degrees of every monomial in every projection polynomial.  For this example, \texttt{sotd} is misled to pick $y \succ x$ (because the resultant of the polynomials in $y$ factors so that as a projection polynomials it has lower \texttt{sotd}, despite the same number of real roots as the resultant in $x$).  

The heuristic \texttt{ndrr} \cite{BDEW13} goes further and counts the size of the induced decomposition of the real line.  It would identify the optimal ordering, but only due to the extra root at $-1.618$ being identified despite not being required for the minimal sign-invariant CAD in $y \succ x$.  
The heuristic based on number of full dimension cells proposed in this paper goes further still by building the full-dimensional cells, but predicts the correct minimal CAD for this problem.

\section{Distribution of cells by dimension}
\label{SEC:Distribution}

\subsection{Distribution for existing problems}

We studied a set of problems from the CAD Example Bank \cite{WBD12_EX}, sourced in turn from the papers  \cite{BH91} and \cite{CMXY09}.  For each problem a sign-invariant CAD was calculated using an implementation of \cite{McCallum1998} in \textsc{Maple} as detailed in \cite{EWBD14}. 
Then for each problem the distribution of cells by dimensions was plotted, as displayed in Fig.~\ref{Formulations:fig:chenBHdimensions}.  
In all such plots (Figures \ref{Formulations:fig:chenBHdimensions}, \ref{Formulations:fig:binomialvsCMXY} and \ref{Formulations:fig:randomCADsandCMXY}) the horizontal axis refers to the cell dimensions and the vertical axis the proportion of the cells in the CAD with those dimensions. 
We see that examples with the same number of variables share similar distributions of cell dimensions: roughly normal but biased towards cells of large dimension.  The closest standard distribution is binomial with a $p$-value $> 0.5$. 

\begin{figure}
\centering
\caption{CAD cell dimension distribution for examples in \cite{WBD12_EX}.  Lines coloured the same relate to problems with the same number of variables (from 2 to 6 going from left to right).
}
\label{Formulations:fig:chenBHdimensions}
\includegraphics[width=0.8\linewidth]{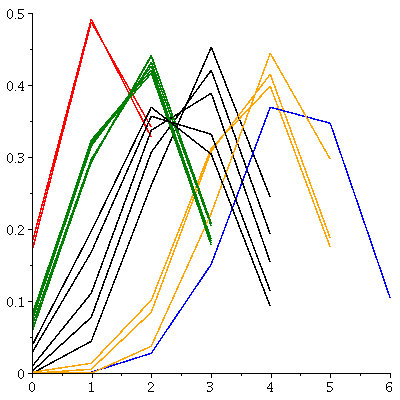}
\end{figure}

\begin{figure}
\caption{Binomial distributions which match the distribution of CAD cells in Figure \ref{Formulations:fig:chenBHdimensions}.  
Calculated from left to right with $(n,p) = [(2,0.6),(3,0.6),(4,0.65),(5,0.7),(6,0.7)]$.}
\label{Formulations:fig:binomialvsCMXY}
\centering
\includegraphics[width=0.8\linewidth]{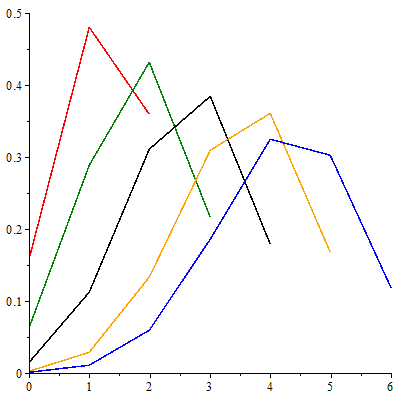}
\end{figure}

Recall that the \textbf{binomial distribution} for $n$ trials with probability $p$ of success is given by:
\begin{equation*}
\mathbb{P}(X = x) := \left\{  
\begin{array}{cc}
\binom{n}{x} p^x (1-p)^{n-x} & 0 \leq x \leq n, \\
0 & {\rm otherwise}.
\end{array}        \right.
\end{equation*}
We determined by eye the $p$-values that best match the examples from \cite{WBD12_EX} and plot these in Fig.~\ref{Formulations:fig:binomialvsCMXY}.  We find that as $n$ increases, the most suitable value of $p$ increases. 

\subsection{Distribution for random cylindrical decompositions}

Our experiments so far suggest the distribution of cell dimensions is largely independent of the individual problems, instead being determined mostly by the number of variables present.  We would like to go further and show they are also independent of the CAD algorithm and implementation used.  

We define \textbf{combinatorially random CADs}.  These are not CADs constructed for randomly generated problems: instead they are hypothetical decompositions of real space made randomly, but with a cylindrical structure.  First a number of variables is chosen at random and then the real line is decomposed into a random number of points and intervals between (a CAD of $\mathbb{R}^1$).  Then each cylinder over a cell is split into a random number of cells: these are assumed sections (zeros of some polynomial) and sectors (the regions in between) but actually here we are constructing just a combinatorial object: a collection of cell indices with no associated projection polynomials. The process is continued until we reach $\mathbb{R}^n$. 

We constructed 45 of these objects in {\sc Maple} using the {\tt rand} command to iteratively build the cell indices (from which cell dimension are easily determined).  Variables were chosen randomly from $\{2,\ldots,6\}$ and cylinders were split using a random number of sections from $\{1,\ldots,7\}$. 
Fig.~\ref{Formulations:subfig:randomCADs} shows the distribution of cell dimensions, and is given alongside the examples from \cite{WBD12_EX} showing the similarity between distributions with the same number of variables.

\subsection{Investigating the combinatorial structure of a CAD}

We now provide some formal justification of the binomial distribution observed in examples. 
In this subsection assume that $\mathcal{D}$ is a CAD or sub-CAD of $\mathbb{R}^n$ and $\mathfrak{D}_i$ the number of cells in $\mathcal{D}$ of dimension $i$ (for $i = 0,\ldots,n$).  We consider the \textbf{induced CADs} of $\mathcal{D}$: the CADs of $\mathbb{R}^j$ ($j = 1, \dots n-1$) formed by projecting $\mathcal{D}$ with respect to the ordering in which it is defined.  The induced CAD of $\mathbb{R}^1$ is a decomposition of the real line into $k_1$ points and $k_1+1$ intervals. 

We make the simplifying assumption that when considering the cylinders over cells from the same induced CAD of $\mathcal{D}$ they are split into the same number of cells.  That is, we assume the cylinder over each cell of $\mathbb{R}^{m}$ consists of $2{k_m}+1$ cells in $\mathbb{R}^{m+1}$ ($k_m$ of which are the same dimension as the base cell).  

\begin{lemma}
\label{Formulations:lem:theoreticalgeneralisation}
Let $\mathcal{D}$ be as described above.  Then 
\begin{equation*}
\mathfrak{D}_{i} = \sum_{\substack{P \subseteq [n] \\ |P| = i}} \left( \prod_{a \in P} (k_a + 1) \prod_{b \in [n] \setminus P} k_b \right),
\end{equation*}
where $[n]$ is the combinatorial shorthand for $\{1,\ldots,n\}$.

In particular we have:
\[
\mathfrak{D}_{0} = \prod_{i=1}^{n} k_i, \qquad 
\mathfrak{D}_{n} = \prod_{i=1}^{n} (k_i + 1).
\]
\end{lemma}

\begin{proof}
The dimension of a cell in $\mathcal{D}$ is equal to the sum of the parity of its cell indices. For a cell to have dimension $i$, it must therefore have $i$ odd indices and $n-i$ even indices.

We can characterise an $i$-dimensional cell by the position of its odd indices. Call the set of these positions $P$. For a fixed $P$ there are many cells associated.  There are a total $k_1+1$ choices of 1-cells if $1 \in P$ and $k_1$ choices of 0-cells if $1 \notin P$. Continuing to build the CAD, at level $j$ there are $k_j + 1$ cell choices if $j \in P$ and $k_j$ choices if $j \notin P$.

Therefore, for a fixed $P$, the number of cells that have an appropriate cell index is given by the product inside the parenthesis.   All that remains is to sum over all possible subsets $P$ of $[n]$ which have cardinality $i$.
\end{proof}

We can see above that in the case where $k_i = k$ for all $i$, then $\{\mathfrak{D}_i\}$ is simply the sequence of binomial coefficients of $(k+(k+1)x)^n$.  In fact, we can see the relationship with the binomial distribution without this extra assumption.

\begin{lemma}
\label{Formulations:lem:generatingfunction}
Let $\mathcal{D}$ be as given in Lemma \ref{Formulations:lem:theoreticalgeneralisation}. Then the generating function for $\mathfrak{D}_i$ is given by:
\begin{equation*}
\prod_{i=1}^n (k_i + (k_i+1)x).
\end{equation*}
That is, $\mathfrak{D}_i$ is the coefficient of $x^i$ in the expansion of the above product.
\end{lemma}

\begin{proof}
Expanding out the product you obtain $x^i$ precisely by choosing $x$ from $i$ different linear factors. This amounts to selecting $i$ integers from the set $[n]$. For a given $P$, each $j \in P$ contributes $k_j+1$ to the coefficient of the generated monomial, and each $j \notin P$ contributes $k_j$. Hence the coefficient of the generated $x^i$ is:
\begin{equation*}
  \prod_{a \in P} (k_a + 1) \prod_{b \in [n] \setminus P} k_b.
\end{equation*}
The coefficient of $x^i$ in the expansion of the product is precisely the summation of all such coefficients:
\begin{equation*}
\sum_{\substack{P \subseteq [n] \\ |P| = i}} \left( 
\prod_{a \in P} (k_a + 1) \prod_{b \in [n] \setminus P} k_b 
\right),
\end{equation*}
which from Lemma \ref{Formulations:lem:theoreticalgeneralisation} is precisely $\mathfrak{D}_i$.
\end{proof}

\begin{figure}
\caption{Comparing the distribution of cell dimensions of combinatorially random CADs with those created for real examples.}
\label{Formulations:fig:randomCADsandCMXY}
\centering
\begin{subfigure}[b]{0.48\columnwidth}
\centering
\includegraphics[width=0.95\linewidth]{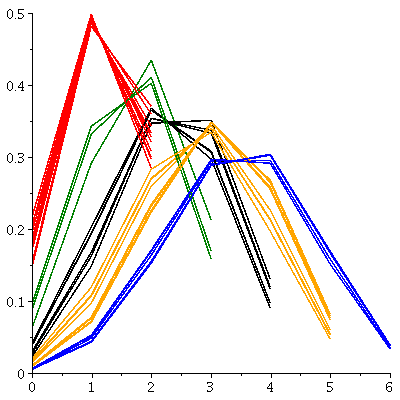}
\caption{Combinatorially random.}
\label{Formulations:subfig:randomCADs}
\end{subfigure}
\begin{subfigure}[b]{0.48\columnwidth}
\centering
\includegraphics[width=0.95\linewidth]{CellDistribution}
\caption{Examples from \cite{WBD12_EX}.}
\label{Formulations:subfig:chentocomparewithrandom}
\end{subfigure}
\end{figure}

\section{Heuristics for CAD problem formulation}
\label{SEC:Heuristic}

As suggested in the introduction, the consistent distribution of cell dimensions can be used for predicting the size of the complete CAD.  Since constructing these cells is far simpler (avoiding all computation with algebraic numbers) we can use it as the basis for heuristics to answer questions of problem formulation.  To experiment with this idea we collated the distributions for our example set \cite{WBD12_EX}, split them by number of variables, and calculated the average proportion of cells that were full-dimensional (shown in Table \ref{Formulation:tab:1layeredexamplebankdist}).  We can then make predictions by comparing the 1-layered CAD for a problem to the average distribution for the number of variables present

\begin{table}[t]
\caption{Average fraction of cells with full-dimensional in CADs for the examples in  \cite{WBD12_EX}.}\label{Formulation:tab:1layeredexamplebankdist}
\centering
\begin{tabular}{ccccc}
Variables & 2 & 3 & 4 & 5 \\
\midrule
Fraction & 0.334 & 0.192 & 0.161 & 0.181
\end{tabular}
\end{table}

\subsection{A new heuristic for choosing a variable ordering}

Example \ref{ex:intro} in the introduction already demonstrated the idea of a heuristic based on full-dimensional cells for picking a variable ordering.  In that example a modest saving was made but more generally the choice of ordering can determine the tractability of a problem.   In \cite{BD07} a class of examples were presented where the ordering changed the complexity from constant to doubly exponential in the number of variables.
The following examples show that a heuristic based on full-dimensional cells works for problems in higher dimensions also, and the potential costs and benefits of using it.

\begin{example}
\label{ex:varord1}
Consider the set of polynomials,
\begin{align*}
F &:= \left\{x^2 + y^2 + z^2 - 1, xy -yz+3, x+y-yz^4\right\}.
\end{align*}
There are six possible orderings of $(x,y,z)$.  We assume all are admissible and seek the one producing the smallest CAD.

For each variable ordering, we create a 1-layered sign-invariant sub-CAD for $F$ using the algorithm in \cite{WBDE14} and a complete sign-invariant CAD for $F$ using the algorithm in \cite{McCallum1998} (both implemented in the \textsc{Maple} package \texttt{ProjectionCAD} \cite{EWBD14}).  The number of cells and computation times for these are recorded, along with the prediction of the number of cells in the full-CAD that would be made by considering the number of full-dimensional cells (multiplying the number of cells in the 1-layered sub-CAD cell count by $\frac{1}{0.192}$).  Table \ref{Formulation:tab:1layeredheuristicExample1} displays these results with the minimal value in each column emboldened.

\begin{table}[b]
\caption{Using 1-layered sub-CADs as a heuristic to pick the variable ordering for a CAD in Example \ref{ex:varord1}.}
\label{Formulation:tab:1layeredheuristicExample1}
\centering
\begin{tabular}{cccccc}
 & \multicolumn{3}{c}{Cells} & \multicolumn{2}{c}{Time} \\
\cmidrule(lr){2-4} 
\cmidrule(lr){5-6}
Order & 1-LCAD & Prediction & CAD & 1-LCAD & CAD \\
\midrule
$x \succ y \succ z$ & {\bf 118} & {\bf 615} & {\bf 539} & {\bf 0.341} & {\bf 2.919} \\
$x \succ z \succ y$ & 160       & 833       & 789       & 0.474       & 5.720       \\
$y \succ x \succ z$ & 340       & 1771      & 1799      & 0.736       & 13.055      \\
$y \succ z \succ x$ & 432       & 2250      & 2211      & 0.950       & 30.638      \\
$z \succ x \succ y$ & 224       & 1167      & 1133      & 0.549       & 10.156      \\
$z \succ y \succ x$ & 392       & 2042      & 2117      & 0.779       & 50.150      \\
\end{tabular}
\end{table}

We see that $x \succ y \succ z$ offers the smallest and quickest 1-layered sub-CAD, and correspondingly, the smallest and quickest complete CAD. The total time for computing all six 1-layered sub-CADs is 
$3.829$ seconds which, along with computing the CAD for $x \succ y \succ z$, means that it would take 
$6.748$ seconds to obtain a CAD using this heuristic.  
This means the heuristic offers a maximum potential saving of 
$43.402$ seconds over computing directly the CAD for $z \succ y \succ x$ and 
1672 cells over the CAD for $y \succ z \succ x$.  
If we had just picked one variable ordering at random then on average our CAD would have taken 
18.773 seconds to compute and have 
1431 cells.  Hence on average the heuristic saves 
892 cells (62\% of the average) and
12.025 seconds (64\% of the average).
\end{example}

\begin{example}
\label{ex:varord2}
Consider the next set of polynomials,
\begin{align*}
\left\{ a^2+b^2+c^2+d^2-1, a^2-4, a-d, b-c, a-c, b-1 \right\},
\end{align*}
in four variables $(a,b,c,d)$ and thus with 24 possible orderings.  
We repeat the experiment detailed in Example \ref{ex:varord2} to produce Table~\ref{Formulation:tab:1layeredheuristicExample2}.

\begin{table}
\caption{Using 1-layered sub-CADs as a heuristic to pick the variable ordering for a CAD in Example \ref{ex:varord2}.}
\label{Formulation:tab:1layeredheuristicExample2}
\centering
\begin{tabular}{cccccc}
 & \multicolumn{3}{c}{Cells} & \multicolumn{2}{c}{Time} \\
\cmidrule(lr){2-4} 
\cmidrule(lr){5-6}
Order                       & 1-LCAD & Prediction & CAD & 1-LCAD & CAD \\
\midrule
$a \succ b \succ c \succ d$ & 7640  & 47453  & 75923  & 10.242 & 387.233 \\
$a \succ b \succ d \succ c$ & 7776  & 48298  & 78187  & 9.212  & 396.658 \\
$a \succ c \succ b \succ d$ & 10644 & 66112  & 106319 & 13.692 & 565.694 \\
$a \succ c \succ d \succ b$ & 11196 & 69540  & 108753 & 13.945 & 557.628 \\
$a \succ d \succ b \succ c$ & 2852  & 17714  & 26903  & 3.477  & 131.893 \\
$a \succ d \succ c \succ b$ & 4340  & 26957  & 41953  & 5.177  & 207.154 \\
$b \succ a \succ c \succ d$ & 6000  & 37267  & 59383  & 9.970  & 316.153 \\
$b \succ a \succ d \succ c$ & 4844  & 30087  & 47879  & 6.755  & 250.658 \\
$b \succ c \succ a \succ d$ & 1946  & 12087  & 18159  & 3.933  & 94.145  \\
$b \succ c \succ d \succ a$ & 1224  & 7602   & 10933  & \textbf{1.676}  & 52.973  \\
$b \succ d \succ a \succ c$ & 1608  & 9988   & 14895  & 2.264  & 73.724  \\
$b \succ d \succ c \succ a$ & 1624  & 10087  & 14595  & 2.174  & 66.825  \\
$c \succ a \succ b \succ d$ & 7684  & 47726  & 72391  & 11.666 & 379.539 \\
$c \succ a \succ d \succ b$ & 6324  & 39279  & 60129  & 8.027  & 308.732 \\
$c \succ b \succ a \succ d$ & 2592  & 16099  & 23705  & 3.509  & 120.485 \\
$c \succ b \succ d \succ a$ & 1404  & 8720   & 12271  & 3.139  & 59.653  \\
$c \succ d \succ a \succ b$ & 3384  & 21018  & 30529  & 4.137  & 142.523 \\
$c \succ d \succ b \succ a$ & 3140  & 19503  & 27545  & 3.811  & 122.740 \\
$d \succ a \succ b \succ c$ & 1184  & 7354   & 10403  & 1.946  & 59.057  \\
$d \succ a \succ c \succ b$ & 1296  & 8050   & 11651  & 1.989  & 62.944  \\
$d \succ b \succ a \succ c$ & 1676  & 10410  & 14927  & 2.628  & 79.229  \\
$d \succ b \succ c \succ a$ & \textbf{1172}  & \textbf{7280}   & \textbf{10213}  & 2.989  & \textbf{51.696}  \\
$d \succ c \succ a \succ b$ & 2364  & 14683  & 21077  & 3.321  & 101.045 \\
$d \succ c \succ b \succ a$ & 1876  & 11652  & 16487  & 2.559  & 76.796  \\
\end{tabular}
\end{table}

The 1-layered sub-CAD cell counts correctly identifies the ordering with the most efficient complete CAD.  However, in this case the timings we would have identified another.  Cell count is likely a more consistent measure for CAD complexity since it avoids many idiosyncrasies of an individual implementation.
The total time to compute all 1-layered sub-CADs for $F$ is 
$132.238$ 
and so using the heuristic to produce a CAD takes 
$183.934$ seconds.
If we picked an ordering at random then on average the CAD would have 
%
27921 cells and take 194.382 seconds.  Hence for this example the heuristic would save a significant number of cells but the time savings would be very modest due to the cost of using the heuristic.

\end{example}

We repeated these experiments on 75 examples each with $\{f_1,f_2,f_3\}$ where $f_i$ are random polynomials in $\{x,y,z\}$ (two quadratic, one linear) generated with {\sc Maple}'s {\tt randpoly} command.
This time we build the 1-layered sub-CADs using the recursive algorithm (Algorithm 4 in \cite{WBDE14}).  This not only constructs a layered sub-CAD, but also a set of unevaluated function calls.  When evaluated these produce both more CAD cells and another set of unevaluated calls.  Combining the new cells with the existing layered sub-CAD will give a sub-CAD with an extra layer (that is, including those cells of one lower dimension).  If we proceed until there are no evaluated calls left then the cells obtained give the complete CAD.
The advantage here is that once an ordering has been selected we do not need to repeat the construction of the full-dimensional cells (giving a further modest saving).  
In the event that more than one ordering had the minimal number of full-dimensional cells we selected the first ordering lexicographically (equivalent to a random choice for these random examples).  

\begin{table}[t]
\caption{Using the number of full dimensional cells to pick the variable ordering for 75 random examples.}
\label{Formulation:tab:recursivelayeredheuristic}
\centering
\begin{tabular}{ccccc}
& \multicolumn{2}{c}{Cells} & \multicolumn{2}{c}{Time} \\
\cmidrule(lr){2-3}
\cmidrule(lr){4-5} 
& Problem Max  & Problem Av & Problem Max & Problem Av \\
\midrule
Average  & $4,719$    & $2,220$ & $64.3$  & $10.0$  \\
Example  & $55.0\%$   & $38.7\%$ & $38.7\%$   & $12.9\%$  \\\midrule
Best    & $12,816$   & $5,204$ & $631.9$ & $143.3$  \\
Example & $93.6\%$   & $84.6\%$ & $84.6\%$   & $70.1\%$ \\\midrule
Worst & $762$      & $297$  & $-44.1$    & $-66.1$   \\
Example      & $13.9\%$   & $6.49\%$  & $-27.9\%$   & $-49.3\%$  \\
\end{tabular}
\end{table}

We found that on average the heuristic saved 38.7\% of the cells and 12.9\% of the computation time for a problem when compared to picking an ordering at random.  However, this masks a lot of variance in the data.  Table \ref{Formulation:tab:recursivelayeredheuristic} gives more details, showing also the examples with the best and worst savings.  We see that even for the worst example the heuristic makes a cell saving, but that for some examples a time saving does not occur (or rather, the saving is outweighed by the cost of the heuristic).  Table \ref{Formulation:tab:recursivelayeredheuristic} compares not only to the average of the different orderings but also to the maximum values, showing the worst cases that can be avoided.
Fig.~\ref{Formulations:fig:recursiveboxplots} shows box plots summarising the 75 examples reinforcing the findings:
\begin{itemize}
\item The cell savings are always positive, meaning the heuristic is an excellent tool for achieving a near minimal CAD.  If the CAD in question is to be computed with extensively in a further application then such cell savings will be of great importance.
\item The time saving can be negative (a cost).  Hence when the aim is to use the heuristic to speed up computations care must be taken.  Further, as the number of variables $n$ increases so too will the time required to compute $n!$ 1-layered sub-CADs and these potential costs.  It is likely that a more appropriate use of the full-dimensional cells from this perspective may be to break ties that result from other (cheaper) measures, or for use when the underlying application (such as quantifier structure) limits the permissible orderings.  
\item When compared to the worst possibility for a problem instead of the problem average of course the savings are greater.  Further, for all but a few outliers the time savings are positive.  Hence the heuristic can be used as a risk-reduction measure.
\end{itemize}

\begin{figure}
\caption{Box plots showing the savings of using the 1-layered sub-CAD heuristic on the 75 random examples.}\label{Formulations:fig:recursiveboxplots}
\centering
\begin{subfigure}[b]{0.48\columnwidth}
\includegraphics[width=0.8\linewidth]{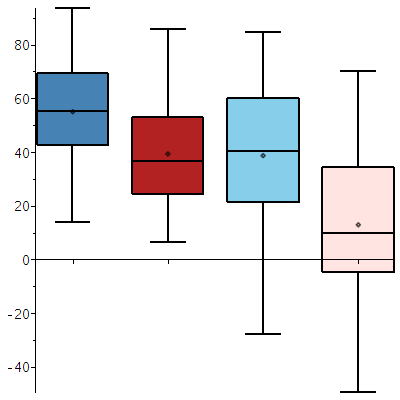}
\caption{Percentage savings in cells.}
\label{Formulations:subfig:recursivecellboxplots}
\end{subfigure}
\begin{subfigure}[b]{0.48\columnwidth}
\includegraphics[width=0.8\linewidth]{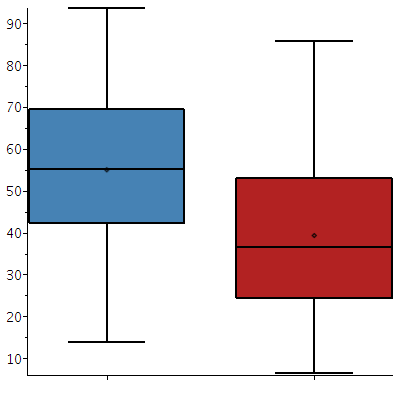}
\caption{Percentage savings in time.}
\label{Formulations:subfig:recursivetimeboxplots}
\end{subfigure}
\end{figure}

\subsection{Algorithms to implement the heuristic}

Algorithm \ref{Formulation:alg:basiclayeredheuristic} demonstrates how the heuristic described above could be implemented efficiently.  We assume a generic CAD input $F$: this was a set of polynomials for the examples above but could more generally be a sequence of formulae say (see \cite{BDEMW13}).  First in step \ref{step:V} we must identify all admissible variable orderings (if there are no restrictions then this is simply all the permutations of the variables defining the polynomials).  Throughout the variable $mc$ stores the minimum number of cells in a computed 1-layered sub-CAD.  The admissible orderings are considered in turn and the full dimensional cells computed.  Here (step \ref{step:1lcad}) an algorithm should be used that is compatible with the required CAD (i.e. same invariance condition).  Ideally it would also be recursive to avoid unnecessary calculation (as described above and in \cite{WBDE14}).  Hence for an ordering $v$ we produce both the 1-layered sub-CAD $L_v$ and a set of unevaluated function calls $U_v$.  If the number of cells computed is a new minimum this is stored.  At the end the variable ordering which contributed the minimal number of full dimensional cells has its layered sub-CAD extended to a complete CAD for the problem (step \ref{step:full}).

\begin{algorithm}
\SetKwInOut{Input}{Input}\SetKwInOut{Output}{Output}
\Input{A CAD input $F$.}
\Output{A CAD for $F$ and a variable ordering $opt$.}
\BlankLine
Set $V$ to be the admissible variable orderings\label{step:V}\;
$mc \leftarrow \infty$\;
\For{$v \in V$}{
    $L_v, U_v \leftarrow {\tt OneLayeredSubCAD}(F,v)$\label{step:1lcad}\;
    \If{$|L_v| < mc$}{
		$mc \leftarrow |L_v|$\;	    
        $opt \leftarrow v$\;
    }
}
$\mathcal{D} \leftarrow {\tt FullCAD}(F,opt,[L_{opt}, U_{opt}])$\label{step:full}\;
\Return $[opt,\mathcal{D}]$\;
\caption{${\tt LayeredHeuristic}$
}\label{Formulation:alg:basiclayeredheuristic}
\end{algorithm}

Algorithm \ref{Formulation:alg:parallellayeredheuristic} describes a (as yet unimplemented) parallel algorithm.  As before we would construct 1-layered sub-CADs for all admissible variable orderings (step 3).  We then await the first to finish and abort the rest (step 6).  We compute the complete CAD for this ordering by evaluating the inert function calls in step 9 (which can also be in parallel).

Algorithm \ref{Formulation:alg:parallellayeredheuristic} chooses the ordering by time taken to compute the full-dimensional cells rather than the number of full-dimensional cells, used by Algorithm \ref{Formulation:alg:basiclayeredheuristic}.  We saw in Example \ref{ex:varord2} that the number of cells can be more closely correlated.  We could modify Algorithm \ref{Formulation:alg:parallellayeredheuristic} to use cells instead by: starting computation of a complete CAD for an ordering only if it had fewer full-dimensional cells than any other computed; if at any point more than two complete CAD computations had been launched, abort the one with fewer full-dimensional cells.

\begin{algorithm}
\SetKwInOut{Input}{Input}\SetKwInOut{Output}{Output}
\Input{A CAD input $F$.}
\Output{A CAD for $F$ in a variable ordering $opt$.}
\BlankLine
Set $V$ to be the admissible variable orderings\;
\ForPar{ {\rm \textbf{each}} $v \in V$\label{step:1lp}}{
     {\rm \textbf{launch}} $L_v, U_v \leftarrow {\tt OneLayCAD}(F,v)$\;
}
Let $v_0$ be the first to finish\;
\ForEach{ $v \in V \setminus \{v_0\}$}{
	{\bf abort} $L_{v}, U_v$\label{step:abort}\;
}
\Repeat{$U_{v_0}$ is empty}{
	\ForPar{ {\rm \textbf{each}} $c \in U_{v_0}$}{
    	 evaluate $c$ and add new cells to $L_{v_0}$ and new unevaluated calls to $U_{v_0}$\;
	}
}
\Return $[v_0,L_{v_0}]$\;
\caption{${\tt Parallel\-Layered\-Heuristic}$
}\label{Formulation:alg:parallellayeredheuristic}
\end{algorithm}

\subsection{Other questions of problem formulation}

We have focused so far on using the new observations on the distribution of CAD cell dimensions as a heuristic for choosing the variable ordering.  However, essentially what we have is a measure of CAD complexity and so we could apply it to other questions of problem formulation for CAD.

\begin{example}
\label{ex:intro2}
Consider again the polynomials $f,g$ 
introduced in Example \ref{ex:intro}.  When considered in Section \ref{SEC:Intro} we built minimal sign-invariant CADs for these polynomials.  However, depending on the underlying application these may provide more information than required.  Suppose that $f$ and $g$ both formed \textbf{equational constraints} (ECs) for the problem: equations whose truth is logically implied by an input formula.  In \cite{Collins1998} a CAD invariant with respect to an EC was defined as a CAD sign-invariant for the polynomial defining an EC and sign-invariant for other polynomials only when that EC is satisfied.  An algorithm to produce such CADs was later presented in \cite{McCallum1999}.  An implementation of this in \texttt{ProjectionCAD} \cite{EWBD14} produced a CAD invariant with respect to $f$ using 21 cells, and one with respect to $g$ using 25 cells.  

The full dimensional cells of these CADs are shown in Fig.~\ref{fig:MVExReturn}.  Using $f$ as the EC creates 8 full-dimensional cells and using $g$ creates 9.  If both are ECs then we can choose to build either of the CADs and a judgement based on the number of full-dimensional cells would lead to the optimal choice.
\end{example}

\begin{figure}
\caption{The image shows CADs built using $f=x-y^2$ and $g=x^2+y^2-1$.  The top image is sign-invariant with respect to $f$ and the second with respect to $g$.  In each case the CAD is also sign-invariant with respect to the other polynomial when the first is zero.}
\label{fig:MVExReturn}
\includegraphics[width=8.8cm]{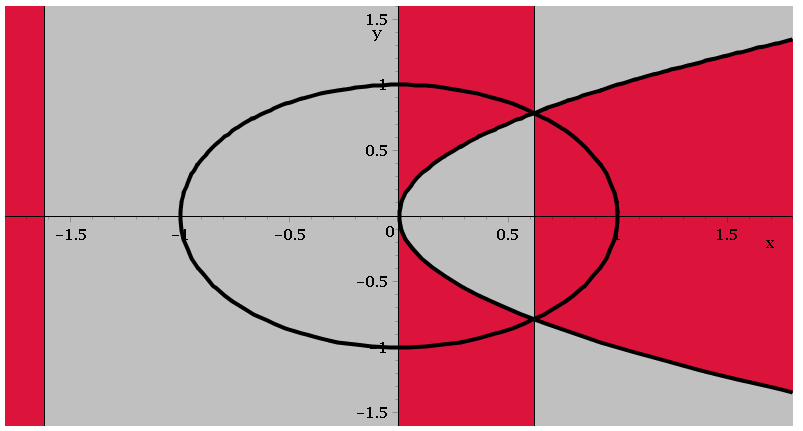}\\
\includegraphics[width=8.8cm]{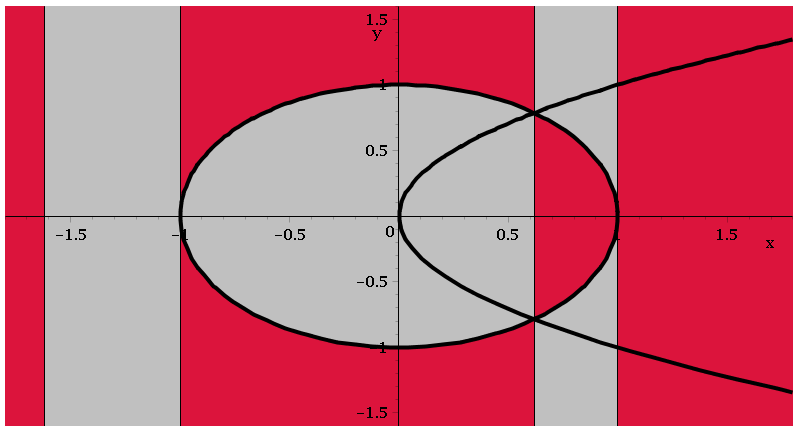}
\end{figure}

\noindent \textbf{Remarks:} 
We note the following about this example:
\begin{enumerate}
\item Of course, when solving a system of equations there are far more efficient techniques to use than CAD.  However, the concept of a CAD with respect to an EC described above is equally applicable to a system of equations and inequalities.
\item Here, both polynomials were ECs but savings were only made from one.  Ideally, we would produce an even simpler CAD where cells are invariant only with the truth of the conjunction of the ECs.  Steps towards this minimal CAD are described in \cite{BCDEMW14} where multiple ECs may be used (following the approach to CAD by regular chains computation originating in \cite{CMXY09}).  In this case there is an analogous questions of problem formulation: the order in which ECs are presented to the algorithm.  This was investigated in \cite{EBCDMW14} where heuristics were developed to help with the problem.  It is likely that the number of full-dimensional cells could also be used to make this choice
\end{enumerate}
 
Other questions of problem formulation for CAD are investigated in \cite{BDEW13} and it is likely that the number of full-dimensional cells could be used as a heuristic for each (with similar caveats on how to use it as outlined above).  We finish by considering a question of problem formulation itself (rather than how a problem is presented to a CAD algorithm).

\begin{example}
\label{ex;piano}
We consider the problem of moving a ladder of length 3 through a right-angled corridor of width 1 (moving from position 1 to position 2 in Fig.~\ref{fig:piano}).  This is an example of a \textbf{piano movers problem} and was first proposed in \cite{Davenport1986}, where it was noted that a CAD could be analysed to give an exact solution.  Of course, a simple analysis shows there is no solution (with it possible only is the ladder is less than $\sqrt{8}$) but we are interested in how this could be decided automatically.

\begin{figure}[ht]
\centering
\begin{tikzpicture}
\draw[ultra thick] (0,0)--(-4,0);
\draw[ultra thick] (0,0)--(0,3);
\draw[ultra thick] (-1,1)--(-4,1);
\draw[ultra thick] (-1,1)--(-1,3);

\draw[orange] (-1.0,0.3)--(-2.8,0.6);
\draw[orange] (-0.6,1.0)--(-0.3,2.8);

\node [below] at (-2.5,0.5) {1};
\node [left] at (-0.5,2.5) {2};

\end{tikzpicture}
\caption{The piano movers problem defined in \cite{Davenport1986}}
\label{fig:piano}
\end{figure}
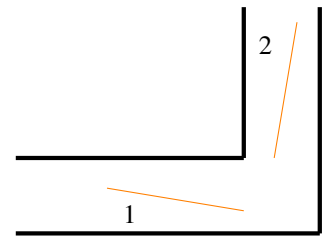

In \cite{Davenport1986} the author described the feasible regions for the ladder, but found that computing a CAD was computationally infeasible, and 28 years later a CAD for this formulation is still very costly.  In \cite{WDEB13} an alternative formulation was proposed: a description of the infeasible region was provided and negated instead.  This simplified the CAD construction considerably.  

We want to identify when such a reformulation of the problem would be beneficial.  In \cite{WDEB13, WBDE14} the process of computing a \textbf{layered variety sub-CAD} is described: as well as restricting to cells of full-dimension we also restrict to cells on a given variety.  The variety was an EC for both formulations (defining the length of the ladder).  The process took around 200 seconds to produce 101,924 cells for the new formulation in \cite{WDEB13}  but timed out for the original one in \cite{Davenport1986}, correctly identifying the tractable formulation for the complete CAD.
\end{example}

\section{Conclusions}

By considering empirical and combinatorial evidence we have shown that the distribution of cells by dimension in a CAD is consistent for a fixed number of variables. This means the size of a CAD can be accurately predicted from the number of full dimensional cells offering new heuristics for CAD problem formulation.  
Such heuristics can be made more efficient by using them with the recursive layered sub-CAD algorithm from \cite{WBDE14} which allows for lazy evaluation of inert computations, avoiding recalculation of results. There is potential for further time savings through parallelism. 

We demonstrated extensively that the number of full-dimensional cells is an effective heuristic for picking the optimal variable ordering for CAD, but that depending on the number of permissible ordering the time savings can be outweighed by the cost of running the heuristic.    
We also demonstrated how the ideas could be used for other questions of problem formulation, with the final example suggesting that the heuristic could be tailored further to the CAD required (using a layered variety sub-CAD when the complete CAD will be invariant with respect to an EC). 

\section*{Acknowledgements}
This work was supported by EPSRC grant: EP/J003247/1.

\bibliographystyle{plain}
\bibliography{CAD}

\end{document}